\documentclass[conference]{IEEEtran}
\IEEEoverridecommandlockouts
\usepackage{amsmath,bm}
\usepackage{bbm}
\usepackage{url}
\usepackage{cite}
\usepackage{tabularx}
\usepackage{colortbl,booktabs}
\usepackage{multirow}
\usepackage{threeparttable}
\usepackage{algorithm}
\usepackage{verbatim}
\usepackage{romannum}
\usepackage{breakurl}

\usepackage{xcolor}
\usepackage{diagbox}
\usepackage{enumerate}
\usepackage{threeparttable}
\usepackage{graphicx}
\usepackage{subfigure}
\usepackage{algpseudocode}
\usepackage{makecell}
\usepackage{nicematrix}
\usepackage{lipsum} % for dummy text

\usepackage{amsthm,amsmath,amssymb}
\usepackage{tablefootnote}

\usepackage{hyperref}

\usepackage{setspace}
\newtheorem{prop}{Proposition}

\newtheorem{definition}{Definition}
\newtheorem{theorem}{Theorem}
\usepackage{optidef}

\usepackage{siunitx}

\def\BibTeX{{\rm B\kern-.05em{\sc i\kern-.025em b}\kern-.08em
    T\kern-.1667em\lower.7ex\hbox{E}\kern-.125emX}}
    
\IEEEpubid{\makebox[\columnwidth]{* These authors contributed equally to this work.\hfill} \hspace{\columnsep}\makebox[\columnwidth]{ }}

\begin{document}

\title{Risk-Aware Value-Oriented Net Demand Forecasting for Virtual Power Plants}

\author{
\IEEEauthorblockN{Yufan Zhang*}
\IEEEauthorblockA{\textit{Electrical and Computer Engineering} \\
\textit{University of California, San Diego}\\
% San Diego, California \\
yuz254@ucsd.edu}
\and
\IEEEauthorblockN{Jiajun Han*}
\IEEEauthorblockA{\textit{Electrical Engineering} \\
\textit{Columbia University}\\
% New York, New York \\
jh4316@columbia.edu}
\and
\IEEEauthorblockN{Yuanyuan Shi}
\IEEEauthorblockA{\textit{Electrical and Computer Engineering} \\
\textit{University of California, San Diego}\\
% San Diego, California \\
yyshi@ucsd.edu}
}

\maketitle

\begin{abstract}
This paper develops a risk-aware net demand forecasting product for virtual power plants, which helps reduce the risk of high operation costs. At the training phase, a bilevel program for parameter estimation is formulated, where the upper level optimizes over the forecast model parameter to minimize the conditional value-at-risk (a risk metric) of operation costs. The lower level solves the operation problems given the forecast. Leveraging the specific structure of the operation problem, we show that the bilevel program is equivalent to a convex program when the forecast model is linear. Numerical results show that our approach effectively reduces the risk of high costs compared to the forecasting approach developed for risk-neutral decision makers. 
\end{abstract}

\begin{IEEEkeywords}
Net demand forecasting, Decision-focused learning, Risk management
\end{IEEEkeywords}

\section{Introduction}
% Buildings contribute a large proportion of global energy use and energy-related emissions []. To reduce the associated emissions, there is a trend to equip buildings with renewable energy resources (RESs). 

With the increasing penetration of Renewable Energy Sources (RESs) in demand side, electricity consumption becomes the net demand, which is calculated by subtracting RES production from the non-dispatchable electricity load. The inherent variability and unpredictability of RESs make the future trend of net demand hard to predict. Improving its forecast accuracy has been the main focus of forecasters. Efforts have been put into developing forecast models by leveraging cutting-edge techniques such as deep learning \cite{alipour2020novel}. 

The predicted net demand is usually used as a parameter in a decision-making problem, such as the energy dispatch problem for scheduling the generators \cite{morales2013integrating}. While forecast accuracy is a useful metric for assessing the statistical quality of forecasts, users are primarily concerned with maximizing the benefits derived from their use, i.e., the decision value. Recent studies show that the statistical quality may not be consistent with the decision value \cite{murphy1993good}. Therefore, research has been conducted on training forecast models aimed at maximizing decision value. Approaches such as integrated optimization \cite{morales2023prescribing}, differentiable programming \cite{donti2017task}, and loss function design \cite{zhang2024improving} have been employed; see the review in \cite{zhang2024improving}. The above studies issue value-oriented forecasts for maximizing the expected decision values. However, such practice takes the perspective of \emph{risk-neutral} decision-makers and fails to protect them from rare but significant losses.

This motivates us the following technical question: \textit{Is it possible to train forecast models to maximize decision values while accounting for the risks arising from uncertainties in decision processes?}

To answer this question, we propose a risk-aware training approach for estimating forecast model parameters for a Virtual Power Plant (VPP) operator. The decision value is defined as minimizing the Day-Ahead (DA) and Real-Time (RT) overall operation cost for the operator. Considering the operator is not risk-neutral, a different objective than the minimization of the expected overall cost is sought at the training phase. Generally speaking, the proper objective for a risk-aware operator penalizes the highest overall operation costs. Therefore, minimizing the Conditional Value-at-Risk (CVaR) (a widely-used metric for quantifying risk \cite{rockafellar2000optimization}) of the overall operation costs is used as the training objective. Using such an objective, a bilevel program is built for parameter estimation, where the upper level optimizes over the forecast model parameter, and the lower level solves the DA and RT operation problems for each sample. We leverage the specific structure of the lower-level operation problems and show that the relationship between the overall operation cost and the forecast is a convex function. With such a function, the bilevel program can be transformed into a convex program when the forecast model is linear, which is efficient to solve by the off-the-shelf solvers. Our main contributions are,

1) Propose a risk-aware forecast model parameter estimation approach at the training phase, which effectively immunizes the operator from high operation costs.

2) Analytically derive a convex function between the forecast and the overall operation cost, enabling a convex parameter estimation program when the forecast model is linear. This ensures global optimality.

\section{Preliminaries}

\subsection{Day-ahead and Real-time Operations of a Virtual Power Plant (VPP) Operator}

In this work, we consider the sequential operation of a VPP operator, which must ensure the satisfaction of net demand, i.e., the demand minus RESs (where the demand can include the load in VPP or the power exchange with others). In DA, the VPP operator determines the energy dispatch of slow-start generators for each time-slot $\tau$ on the next day $d$. Such decisions are made based on the net demand forecast $\hat{y}_{d,\tau}$ at time $t$ on day $d-1$, before knowing true net demand. After the net demand realization $y_{d,\tau}$ is revealed, the operator is responsible for deficit or surplus of power imbalance at each time-slot $\tau$ on the next day $d$. The operator can compensate for the imbalance, by resorting to the flexible resources. 
% Since operator is too small to affect alone market clearing price, we assume,
% \begin{assumption}
%     operator is a price-taker in DA and RT markets.
% \end{assumption}
Next, we detail the DA and RT operations.
 % Adjust this value to reduce the space above the figure
% \begin{figure}[h]
%   \centering
%   \includegraphics[scale=0.45]{Figures/timeline.pdf}
%   \caption{\scriptsize An illustration of the sequential operation.}
%   \label{sequential}
% \end{figure}
 % Adjust this value to reduce the space below the figure
\vspace{-0.2em}
\subsubsection{DA Operation} For any time-slot $\tau$ on the next day $d$, the VPP operator schedules the slow-start generators to minimize the generation cost. We assume that the inter-temporal constraints, such as ramping constraints, are not explicitly considered \cite{morales2023prescribing}. The DA operation problem is,
\begin{subequations}\label{1}
\begin{alignat}{2} 
&\mathop{\min}_{\bm{p}^{\text{DA}}_{d,\tau}}
  &&\ \bm{\rho}^{\text{DA},\top}\bm{p}^{\text{DA}}_{d,\tau}\\ 
& \text{s.t.}
&&    \bm{1}^\top\bm{p}^{\text{DA}}_{d,\tau}=\hat{y}_{d,\tau}:\lambda_{d,\tau}\label{1(b)}
    \\ 
    &&&  \bm{0} \leq \bm{p}^{\text{DA}}_{d,\tau} \leq \overline{\bm{p}}^{\text{DA}}:\underline{\bm{\nu}}_{d,\tau},\overline{\bm{\nu}}_{d,\tau} \label{1(c)},
\end{alignat}
\end{subequations}
where $\bm{p}_{d,\tau}^{\text{DA}}$ denotes the scheduling decisions. $\bm{\rho}^{\text{DA}}$ is the marginal generation cost and $\overline{\bm{p}}^{\text{DA}}$ is the generation upper limit. The dual variables are listed after the colons. The dual problem of \eqref{1} is,
\begin{subequations}\label{1RT}
\begin{alignat}{2} 
&\mathop{\max}_{\lambda_{d,\tau},\underline{\bm{\nu}}_{d,\tau},\overline{\bm{\nu}}_{d,\tau}}
  &&\ \lambda_{d,\tau}\hat{y}_{d,\tau}-\overline{\bm{\nu}}_{d,\tau}^\top\overline{\bm{p}}^{\text{DA}}\label{1RT(a)}\\ 
& \text{s.t.}
&&    \bm{\rho}^{\text{DA}}-\lambda_{d,\tau}\bm{1}+\overline{\bm{\nu}}_{d,\tau}-\underline{\bm{\nu}}_{d,\tau}=0\label{1RT(b)}
    \\ 
    &&&  \underline{\bm{\nu}}_{d,\tau},\overline{\bm{\nu}}_{d,\tau}\geq0 \label{1RT(c)},
\end{alignat}
\end{subequations}

\subsubsection{RT Operation} The deviation between the net demand realization $y_{d,\tau}$ and the DA forecast $\hat{y}_{d,\tau}$ is inevitable. The operator is required to compensate the deviation $y_{d,\tau}-\hat{y}_{d,\tau}$ as time close to the delivery. When energy deficit happens, i.e., $y_{d,\tau}-\hat{y}_{d,\tau} > 0$, the flexible resources output power $\bm{p}^{+}_{d,\tau}$ to compensate the deficit, at the marginal cost of $\bm{\rho}^{+}$. When energy surplus happens, i.e., $y_{d,\tau}-\hat{y}_{d,\tau} < 0$, the flexible resources absorb the excessive electricity $\bm{p}^{-}_{d,\tau}$ at the marginal utility of $\bm{\rho}^{-}$. Minimizing the cost of balancing the deviation, the operator solves RT operation problem at each time-slot $\tau$ on day $d$ as,
\begin{subequations}\label{2}
\begin{alignat}{2} 
&\mathop{\min}_{\bm{p}^{+}_{d,\tau},\bm{p}^{-}_{d,\tau}}
  &&\bm{\rho}^{+,\top}\bm{p}^{+}_{d,\tau}-\bm{\rho}^{-,\top}\bm{p}^{-}_{d,\tau}\\ 
& \text{s.t.}
&&    \bm{1}^\top(\bm{p}^{+}_{d,\tau}-\bm{p}^{-}_{d,\tau})=y_{d,\tau}-\hat{y}_{d,\tau}:\gamma_{d,\tau}\label{2(b)}
    \\ 
    &&&  \bm{0} \leq \bm{p}^{-}_{d,\tau} \leq \overline{\bm{p}}^{-}:\underline{\bm{\mu}}_{d,\tau},\overline{\bm{\mu}}_{d,\tau} \label{2(d)}\\
    &&&  \bm{0} \leq \bm{p}^{+}_{d,\tau} \leq \overline{\bm{p}}^{+}:\underline{\bm{\eta}}_{d,\tau},\overline{\bm{\eta}}_{d,\tau} \label{2(e)},
\end{alignat}
\end{subequations}
where \eqref{2(b)} ensures that the deviation is settled. The adjustment provided by flexible loads is bounded by technical limits via \eqref{2(d)}-\eqref{2(e)}, where $\overline{\bm{p}}^{+},\overline{\bm{p}}^{-}$ are the upper limits respectively. After solving \eqref{2}, the optimal RT solutions $\bm{p}^{+,*}_{d,\tau},\bm{p}^{-,*}_{d,\tau}$ are obtained. The dual variables of the RT operation are listed after the colons in \eqref{2(b)}-\eqref{2(e)}. The dual problem of \eqref{2} is,
\begin{subequations}\label{dualRT}
\begin{alignat}{2} 
&\mathop{\max}_{\Xi}
  &&\ \gamma_{d,\tau}(y_{d,\tau}-\hat{y}_{d,\tau})-\overline{\bm{\mu}}_{d,\tau}^\top\overline{\bm{p}}^{-}-\overline{\bm{\eta}}_{d,\tau}^\top\overline{\bm{p}}^{+}\label{dualRTa}\\ 
& \text{s.t.}
&&    
-\bm{\rho}^{-}+\gamma_{d,\tau}\bm{1}+\overline{\bm{\mu}}_{d,\tau}-\underline{\bm{\mu}}_{d,\tau}=0\label{dualRTc}\\
    &&&  \bm{\rho}^{+}-\gamma_{d,\tau}\bm{1}+\overline{\bm{\eta}}_{d,\tau}-\underline{\bm{\bm{\eta}}}_{d,\tau}=0 \label{dualRTd}\\
    &&& \underline{\bm{\mu}}_{d,\tau},\overline{\bm{\mu}}_{d,\tau},\underline{\bm{\eta}}_{d,\tau},\overline{\bm{\eta}}_{d,\tau} \geq 0,
\end{alignat}
\end{subequations}
where $\Xi=\{\gamma_{d,\tau},\underline{\bm{\mu}}_{d,\tau},\overline{\bm{\mu}}_{d,\tau},\underline{\bm{\eta}}_{d,\tau},\overline{\bm{\eta}}_{d,\tau}\}$. The optimal solution of \eqref{dualRT} is $\{\gamma^*_{d,\tau},\underline{\bm{\mu}}^*_{d,\tau},\overline{\bm{\mu}}^*_{d,\tau},\underline{\bm{\eta}}^*_{d,\tau},\overline{\bm{\eta}}^*_{d,\tau}\}$.

Here, we define the overall operation cost of the VPP operator in both DA and RT  at a time-slot $\tau$,
\begin{definition}
    We define the overall operation cost of the operator at time-slot $\tau$ on day $d$ as,
\begin{subequations}\label{overallcost}
\begin{align}
        &\overbrace{\bm{\rho}^{\text{DA},\top}\bm{p}^{\text{DA},*}_{d,\tau}}^\text{DA primal optimal objective}+\overbrace{\bm{\rho}^{+,\top}\bm{p}^{+,*}_{d,\tau}-\bm{\rho}^{-,\top}\bm{p}^{-,*}_{d,\tau}}^\text{RT primal optimal objective}\label{overallcosta}\\
        =&\overbrace{\lambda_{d,\tau}^{*}\cdot \hat{y}_{d,\tau}-\overline{\bm{\nu}}_{d,\tau}^{*,\top}\overline{\bm{p}}^{\text{DA}}}^\text{DA dual optimal objective}+\nonumber\\
        &\overbrace{\gamma^*_{d,\tau}(y_{d,\tau}-\hat{y}_{d,\tau})-\overline{\bm{\mu}}_{d,\tau}^{*,\top}\overline{\bm{p}}^{-}-\overline{\bm{\eta}}_{d,\tau}^{*,\top}\overline{\bm{p}}^{+}}^\text{RT dual optimal objective}\label{overallcostb},
    \end{align}\end{subequations}
\end{definition}
where \eqref{overallcosta} is the DA primal optimal objective plus RT primal optimal objective, and \eqref{overallcostb} is the DA dual optimal objective plus RT dual optimal objective. The two objectives are equal due to the strong duality theorem \cite{boyd2004convex}.
\eqref{overallcosta},\eqref{overallcostb} are related with the forecast $\hat{y}_{d,\tau}$ and the realization $y_{d,\tau}$.
\vspace{0em}

\subsection{Optimization of Conditional Value-at-Risk}

Let $\ell(x,Y)$ denote the cost function associated with the decision $x \in \mathcal{X}$ (where $\mathcal{X}$ is the constraint set) and random variable $Y$. For a specific $x \in \mathcal{X}$, $\ell(x,Y)$ is a random variable. When the realization of $Y$ is $y$. the realization of $\ell(x,Y)$ is $\ell(x,y)$. Let $\text{Var}_\beta$ denote the $\beta$-quantile of the cost $\ell(x,Y)$. The corresponding $\text{CVaR}_\beta$ is defined as the expected value of the costs that are larger or equal to $\text{VaR}_\beta$. The illustration of $\text{CVaR}_\beta$ is shown in Fig. \ref{cvar}. The goal is to find an optimal solution of $x$ to minimize $\text{CVaR}_\beta$. \cite{rockafellar2000optimization} shows that such a goal can be achieved via solving an optimization program,
\begin{equation}
    x^*,\alpha^*=\mathop{\arg\min}_{x \in \mathcal{X},\alpha}\alpha+\frac{1}{(1-\beta)M}\sum_{m=1}^M [\ell(x,y_m)-\alpha]^+,
\end{equation}
where $\{y_m\}_{m=1}^M$ are the $M$ scenarios sampled from the probability distribution of $Y$. The expresssion $[\ell(x,y_m)-\alpha]^+=\ell(x,y_m)-\alpha$ when $\ell(x,y_m) > \alpha$ and $[\ell(x,y_m)-\alpha]^+=0$ when $\ell(x,y_m) \leq \alpha$. If the cost function $\ell(x,y_m)$ is convex, the objective $\alpha+\frac{1}{(1-\beta)M}\sum_{m=1}^M [\ell(x,y_m)-\alpha]^+$ is convex \cite{rockafellar2000optimization}. The optimal solution $x^*$ minimizes the $\text{CVaR}_\beta$, and $\alpha^*$ gives the corresponding $\text{Var}_\beta$.

\begin{figure}[t]
  \centering
  % Requires \usepackage{graphicx}
  \includegraphics[scale=0.3]{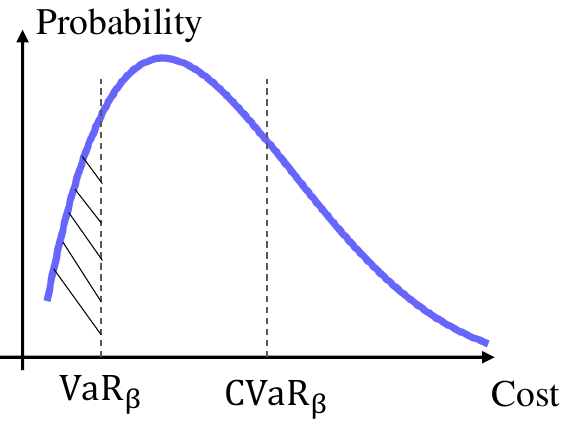}
  \caption{\scriptsize The illustration of $\text{CVaR}_\beta$. The dashed area measures $\beta$. $\text{VaR}_\beta$ is the $\beta$-quantile of the cost. $\text{CVaR}_\beta$ is the expected value of the cost above $\text{VaR}_\beta$.}
\label{cvar}
\end{figure}

\section{Risk-Aware Parameter Estimation}
In this section, we estimate the forecast model parameter to minimize the overall operation cost of the operator while considering risk. Let $g(\ \cdot \ ;\Theta)$ denote the forecast model with the parameter $\Theta$, and $\bm{s}_{d,\tau}$ denote the context. The net demand forecast for the time-slot $\tau$ on day $d$ is issued by,
\begin{equation}\label{fm}
    \hat{y}_{d,\tau}= g(\bm{s}_{d,\tau};\Theta)
\end{equation}

% \begin{figure}[t]
%   \centering
%   % Requires \usepackage{graphicx}
%   \includegraphics[scale=0.5]{Figures/bilevel.pdf}
%   \caption{An illustration of the bilevel program for parameter estimation.}
% \label{bilevel}
% \end{figure}

Given the dataset $\{\{\bm{s}_{d,\tau},y_{d,\tau}\}_{\tau=1}^T\}_{d=1}^D$, comprising historical context $\bm{s}_{d,\tau}$ and net demand realization $y_{d,\tau}$ over $D$ days, we formulate a bilevel program, where the upper level optimizes over the forecast model parameter aware of the risk of high operation costs. The lower level solves the DA and RT operation problems. The bilevel program is as follows,
\begin{subequations}\label{bilevel}
\begin{alignat}{2} 
&\underset{{\Theta},\alpha}{\mathop{\min}}&& \  \alpha+\frac{1}{D\cdot T\cdot(1-\beta)}\sum_{d=1}^{D}\sum_{\tau=1}^{T}[\eqref{overallcosta}-\alpha]^+\label{bilevel_a}\\
    & \text{s.t.} 
    && \hat{y}_{d,\tau}=g(\bm{s}_{d,\tau};\Theta),\forall \tau=1,...,T,\forall d=1,...,D\label{bilevel_b}\\     &&&\begin{rcases}\eqref{1},\forall \tau=1,...,T,\forall d=1,...,D\\
    \eqref{2},\forall \tau=1,...,T,\forall d=1,...,D\end{rcases}\text{Lower level}\label{bilevel_e},
\end{alignat}
\end{subequations}

To reduce the risk of high costs, the upper level optimizes over the forecast model parameter $\Theta$ toward minimizing $\text{CVaR}_\beta$ of the overall operation cost in the training set, as stated in \eqref{bilevel_a}. This overall operation cost, given in \eqref{overallcosta}, is informed by the lower-level solutions in \eqref{bilevel_e}.  The degree of risk aversion is controlled by $\beta$. When the operator is more risk-averse regarding high costs, $\beta$ is set to a larger value. We have the following proposition for the extreme case where $\beta=0$.
\begin{prop}\label{prop1}
    When the operator sets $\beta=0$, the operator is risk-neutral. The forecast model parameter is trained toward minimizing the expected overall operation cost. 
\end{prop}
\begin{proof}
    When $\beta=0$, $\alpha$ is the minimum of the overall operation cost \eqref{overallcost} in the test set. Therefore, $\forall \tau=1,...,T,d=1,...,D$, $[\eqref{overallcost}-\alpha]^+=\eqref{overallcost}-\alpha$. Then \eqref{bilevel_a} is equivalent with $\frac{1}{D\cdot T}\sum_{d=1}^{D}\sum_{\tau=1}^{T}\eqref{overallcost}$, i.e., the expected operation cost.
\end{proof}

Proposition \ref{prop1} shows that when $\beta=0$, \eqref{bilevel_a} is equivalent to the expected overall operation cost. The bilevel program for parameter estimation in \eqref{bilevel} is therefore reduced to the one proposed in \cite{zhang2024improving} for risk-neutral operators.

To enhance the connection between the forecast and the overall operation cost, we replace \eqref{overallcosta} with \eqref{overallcostb}, and substitute the lower-level RT primal problem with its dual counterpart. The bilevel program in \eqref{bilevel} is reformulated as,
\begin{subequations}\label{Dbilevel}
\begin{alignat}{2} 
&\underset{{\Theta},\alpha}{\mathop{\min}}&& \  \alpha+\frac{1}{D\cdot T\cdot(1-\beta)}\sum_{d=1}^{D}\sum_{\tau=1}^{T}[\eqref{overallcostb}-\alpha]^+\label{Dbilevel_a}\\
    & \text{s.t.} \notag
    && \eqref{bilevel_b}\\
    &&&\begin{rcases}\eqref{1RT},\forall \tau=1,...,T,\forall d=1,...,D\\
    \eqref{dualRT},\forall \tau=1,...,T,\forall d=1,...,D\end{rcases}\text{Lower level}\label{Dbilevel_e},
\end{alignat}
\end{subequations}

Solving \eqref{Dbilevel} remains challenging. A common solution approach involves substituting the lower level with its KKT conditions, which are converted to mixed-integer constraints. The number of mixed-integer constraints grows as the number of training set samples increases, rendering the hard-to-solve large-scale mixed-integer program. In the upcoming section, we'd like to suggest deriving an analytical function linking the lower-level dual solutions and the forecast. This approach allows for the elimination of the lower level. Substituting this function into \eqref{overallcostb} yields the relationship between the forecast and the overall cost. We will then demonstrate the properties of this function, which aid in solving \eqref{Dbilevel}.

\section{Deriving the Function between the Forecast and the Overall Operation Cost}
 In this section, we derive the analytical relationship between the forecast $\hat{y}_{d,\tau}$ and the overall operation cost as given in \eqref{overallcostb}. Given the involvement of dual solutions in \eqref{overallcostb}, we first derive the analytical function between the optimal dual solutions $\{\lambda_{d,\tau}^{*},\overline{\bm{\nu}}_{d,\tau}^*,\gamma^*_{d,\tau},\overline{\bm{\mu}}^*_{d,\tau},\overline{\bm{\eta}}^*_{d,\tau}\}$ and the forecast $\hat{y}_{d,\tau}$. These relationships are determined by the DA and RT dual problems in \eqref{1RT},\eqref{dualRT}, which are linear programs where the forecast $\hat{y}_{d,\tau}$ serves as a parameter. For general linear programs \eqref{lp} (where $\bm{x}$ is the primal variable, $\bm{c},\bm{G},\bm{\psi},\bm{F}$ are the fixed parameters, $\bm{\omega}$ is the parameter of interest) and its dual problem \eqref{dualG} (where $\bm{\sigma}$ is the dual variable), \cite{borrelli2003geometric} outlines the properties of the functions linking the optimal dual solutions and the optimal objective to the parameters.
\begin{subequations}\label{lp}
\begin{alignat}{2}
 \bm{x}^*=&\mathop{\arg\min}_{\bm{x}} && \quad \bm{c}^\top  \bm{x}\label{lpa}\\
    &\text{s.t.} &&  \quad  \bm{G}\bm{x}\leq \bm{\psi}+\bm{F}\bm{\omega}:\bm{\sigma}\label{lpb}.
\end{alignat}
\end{subequations}

\begin{equation}\label{dualG}
     \bm{\sigma^*}=\mathop{\arg\max}_{\bm{\sigma} \geq 0} -\bm{\sigma}^\top(\bm{\psi}+\bm{F}\bm{\omega}) 
 \end{equation}

 \begin{theorem}\label{theom1}
     \cite{borrelli2003geometric} Consider general linear programs \eqref{dualG}. The function, which describes the change in the optimal dual solutions $\bm{\sigma^*}$ as the parameter changes $\bm{\omega}\in \Omega$ (where $\Omega$ is a convex polytope), is a stepwise function. The relationship between the optimal objective $-\bm{\sigma}^{*\top}(\bm{\psi}+\bm{F}\bm{\omega})$ and the parameter $\bm{\omega}\in \Omega$ is a piecewise linear and convex function. 
 \end{theorem}

Theorem \ref{theom1} shows that there exists a set of polyhedral partitions $R_1,...,R_N$ of $\Omega$, $\forall \bm{\omega} \in R_i$, the function between the optimal dual solution and $\bm{\omega}$ is a constant function, whose value $\bm{\sigma}^{*i}$ can be obtained by solving the dual problem in \eqref{dualG} given any $\bm{\omega} \in R_i$. The corresponding linear function between the optimal objective and the parameter $\bm{\omega} \in R_i$ can be obtained, which is $-\bm{\sigma}^{*i,\top}(\bm{\psi}+\bm{F}\bm{\omega})$. \cite{borrelli2003geometric} shows how to divide $\Omega$ into a set of partitions $R_1,...,R_N$ via the active constraints in the primal problem  \eqref{lp} in a general sense. For the specific DA operation problem in \eqref{1},\eqref{1RT} and the RT operation problem in \eqref{2},\eqref{dualRT}, we leverage the structure to obtain polyhedral partitions and subsequently obtain the functions that link the optimal overall cost to the parameters.
\begin{figure}[t]
  \centering
  % Requires \usepackage{graphicx}
  \includegraphics[scale=0.49]{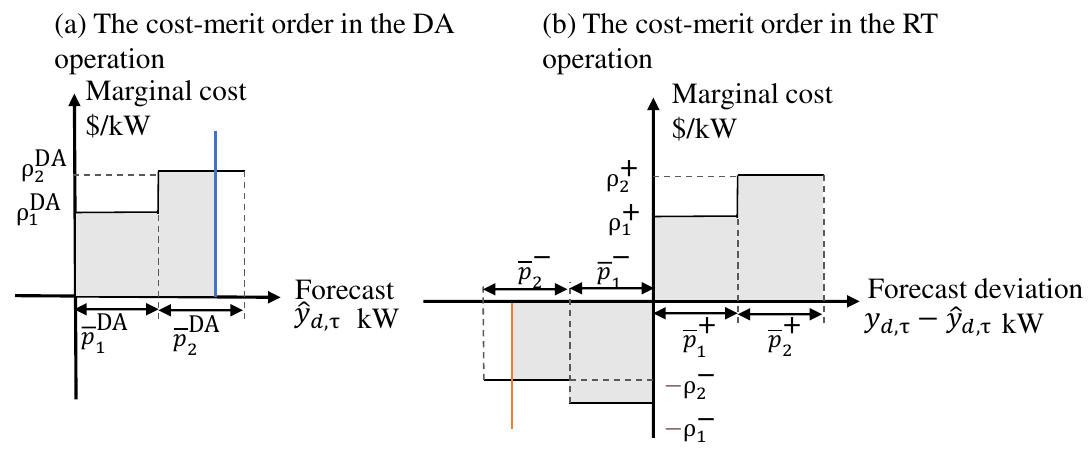}
  \caption{\scriptsize Illustration of a cost-merit order dispatch. (a) In the DA operation, there are two generators with marginal costs $\rho_1^{\text{DA}},\rho_2^{\text{DA}}$ and capacities $\overline{p}_1^{\text{DA}},\overline{p}_2^{\text{DA}}$. The blue solid line represents the net demand forecast $\hat{y}_{d,\tau}$. Since the generator 1 is the cheapest one, it is dispatched to full capacity. The generator 2 is partially dispatched until the net demand is settled. (b) In the RT operation, flexible resources 1 and 2 addressing energy deficit with marginal cost $\rho_1^{+},\rho_2^{+}$ and capacities $\overline{p}_1^{+},\overline{p}_2^{+}$. Flexible resources 3 and 4 address energy surplus with marginal utility $\rho_1^{-},\rho_2^{-}$ and capacities $\overline{p}_1^{-},\overline{p}_2^{-}$. 
  % The blue solid line represents the positive forecast deviation $y_{d.\tau}-\hat{y}_{d,\tau} > 0$. Since the flexible resource 1 is the cheapest one, it is dispatched to full capacity. The flexible resource 2 is partially dispatched until the deviation is settled. 
  The orange solid line represents the negative forecast deviation $y_{d,\tau}-\hat{y}_{d,\tau} < 0$. Since flexible resource 3 has a higher marginal utility, it is dispatched to the full capacity. The flexible resource 4 is partially dispatched until the deviation is settled.}
\label{costmerit}
\end{figure}

\subsection{Obtaining DA and RT Partitions}

We leverage the structure of the primal DA and RT problems in \eqref{1},\eqref{2} to obtain the polyhedral partitions. \eqref{1},\eqref{2} stand for economic dispatch problems, where the resources are dispatched in a cost-merit order. For instance, in the DA operation \eqref{1}, the solutions are related to the forecast $\hat{y}_{d,\tau}$. The generator with a lower marginal cost is dispatched first. In the RT operation \eqref{2}, the solutions are related to the forecast deviation $y_{d,\tau}-\hat{y}_{d,\tau}$. When the operator faces an energy deficit, i.e.,  
 $y_{d,\tau}-\hat{y}_{d,\tau} >0$, and flexible resources output electricity, the resources with lower marginal cost are resorted first. When the operator faces an energy surplus, i.e., $y_{d,\tau}-\hat{y}_{d,\tau} <0$, and flexible resources absorb electricity, the resources with higher marginal utility are resorted first. This is illustrated in Fig. \ref{costmerit}. The cost-merit order indicates that the partitions are related to the order of the marginal cost.

 To ease the discussion that follows, 
 % we group the decisions and prices in the case of energy deficit into vectors $\bm{p}^{-}_{d,\tau}$ (consisting of bought electricity from RT market $p_{d,\tau}$ and curtail electricity $\bm{p}^{-}_{d,\tau}$) and $\bm{\lambda}^{-}$ (consisting of the corresponding prices $\lambda^{\text{RT}}$ and $\bm{\lambda}^{-}$). The corresponding upper bounds are grouped in  $\overline{\bm{p}}^{-}$.
 we consider the generator dispatch in $\bm{p}^{\text{DA}}_{d,\tau},\overline{\bm{p}}^{\text{DA}}$ are ordered, such that $i \leq j$ if and only if $\rho_i^{\text{DA}} \leq \rho_j^{\text{DA}}$, where $\rho_i^{\text{DA}},\rho_j^{\text{DA}}$ are the $i_{th}$ and $j_{th}$ elements in $\bm{\rho}^{\text{DA}}$. Additionally, resources in $\bm{p}^{+}_{d,\tau}$, $\overline{\bm{p}}^{+}$ are ordered, such that $i \leq j$ if and only if $\rho_i^{+} \leq \rho_j^{+}$, where $\rho_i^{+},\rho_j^{+}$ are the $i_{th}$ and $j_{th}$ elements in $\bm{\rho}^{+}$. Likewise,
 % we group the decisions and prices in the case of energy surplus into vectors $\bm{p}^{+}_{d,\tau}$ (consisting of sold electricity $p_{d,\tau}$,$\bm{p}^{+}_{d,\tau}$ to RT market and flexible loads) and $\bm{\lambda}^{+}$ (consisting of the corresponding prices $\lambda^{\text{RT}}$ and $\bm{\lambda}^{+}$). The corresponding upper bounds are grouped in  $\overline{\bm{p}}^{+}$. 
 the resources in $\bm{p}^{-}_{d,\tau}$, $\overline{\bm{p}}^{-}$ are ordered, such that $i \leq j$ if and only if $\rho_i^{-} \geq \rho_j^{-}$, where $\rho_i^{-},\rho_j^{-}$ are the $i_{th}$ and $j_{th}$ elements in $\bm{\rho}^{-}$. 

 Denote $|\bm{x}|$ as the dimension of vector $\bm{x}$. In the DA operation problem, there are $|\overline{\bm{p}}^{\text{DA}}|$ partitions. The first DA partition is,
 \begin{equation}\label{DAp1}
     0 \leq \hat{y}_{d,\tau} \leq \overline{p}^{\text{DA}}_1
 \end{equation}
 When the forecast belongs to the first partition, the cheapest generator is dispatched. The $o_{th} > 1$ DA partition is,
 \begin{equation}\label{DApo}
     \sum_{k=1}^{o-1}\overline{p}_{k}^{\text{DA}} \leq \hat{y}_{d,\tau} \leq \sum_{k=1}^{o}\overline{p}_{k}^{\text{DA}}
 \end{equation}
 
 In the RT operation problem, when $y_{d,\tau}-\hat{y}_{d,\tau} >0 $, there are $|\overline{\bm{p}}^{+}|$ polyhedral partitions. The first partition is,
\begin{equation}\label{10}
0 \leq y_{d,\tau}-\hat{y}_{d,\tau} \leq \overline{p}_{1}^{+}
\end{equation}

 The $i_{th} > 1$ partition is as follows,
\begin{equation}\label{11}   
\sum_{k=1}^{i-1}\overline{p}_{k}^{+} \leq y_{d,\tau}-\hat{y}_{d,\tau} \leq \sum_{k=1}^{i}\overline{p}_{k}^{+}
\end{equation}

Likewise, when  $y_{d,\tau}-\hat{y}_{d,\tau} <0 $, there are $|\overline{\bm{p}}^{-}|$ polyhedral partitions. The first partition is,
\begin{equation}\label{12}
-\overline{p}_{1}^{-} \leq 
y_{d,\tau}-\hat{y}_{d,\tau} \leq 0
\end{equation}
 
 The $j_{th} >1$ partition is as follows,
\begin{equation}\label{13}
-\sum_{k=1}^{j}\overline{p}_{k}^{-} \leq y_{d,\tau}-\hat{y}_{d,\tau} \leq -\sum_{k=1}^{j-1}\overline{p}_{k}^{-}
\end{equation}

To sum up, the DA partitions are in \eqref{DAp1} and \eqref{DApo}, whose number is $|\overline{\bm{p}}^{\text{DA}}|$. The RT partitions are in \eqref{10}-\eqref{13}, whose number is $|\overline{\bm{p}}^{+}|+|\overline{\bm{p}}^{-}|$. To ease the discussion, we group the DA partitions in \eqref{DAp1} and \eqref{DApo} into the set $\{R_o^{\text{DA}}\}_{o=1}^{|\overline{\bm{p}}^{\text{DA}}|}$, where $R_o^{\text{DA}}$ is a DA partition, and we group the RT partitions in \eqref{10}-\eqref{13} into the set $\{R_n^{\text{RT}}\}_{n=1}^{|\overline{\bm{p}}^{+}|+|\overline{\bm{p}}^{-}|}$, where $R_n^{\text{RT}}$ is a RT partition.

\subsection{Function between the Forecast and the Overall Cost} 
For each of the partition $R_o^{\text{DA}}$ in $\{R_o^{\text{DA}}\}_{o=1}^{|\overline{\bm{p}}^{\text{DA}}|}$, $R_n^{\text{RT}}$ in $\{R_n^{\text{RT}}\}_{n=1}^{|\overline{\bm{p}}^{+}|+|\overline{\bm{p}}^{-}|}$, we sample  $\hat{y}_{d,\tau} \in R_o^{\text{DA}}$, $y_{d,\tau}-\hat{y}_{d,\tau} \in R_n^{\text{RT}}$ from $R_o^{\text{DA}}$ and $R_n^{\text{RT}}$. Then, we plug the sampled $\hat{y}_{d,\tau} \in R_o^{\text{DA}}$ into the DA dual problem \eqref{1RT}, and obtain the DA dual solutions $\{\lambda_{d,\tau}^{*o},\overline{\bm{\nu}}_{d,\tau}^{*o}\}$, which are the output of the constant function between the optimal DA dual solutions and forecast $\hat{y}_{d,\tau}$. After that, we plug the sampled $y_{d,\tau}-\hat{y}_{d,\tau} \in R_n^{\text{RT}}$ into the RT dual problem \eqref{dualRT}, and obtain the RT dual solutions $\{\gamma^{*n}_{d,\tau},\overline{\bm{\mu}}^{*n}_{d,\tau},\overline{\bm{\eta}}^{*n}_{d,\tau}\}$,  which are the output of the constant function between the optimal RT dual solutions and forecast deviation $y_{d,\tau}-\hat{y}_{d,\tau}$. By plugging the optimal dual solutions into the dual objectives \eqref{1RT(a)}, the linear function between the optimal DA objective and $\hat{y}_{d,\tau} \in R_o^{\text{DA}}$ can be obtained,
\begin{equation}\label{DAdualobjfunc}
\lambda_{d,\tau}^{*o}\cdot \hat{y}_{d,\tau}-\overline{\bm{\nu}}_{d,\tau}^{*o,\top}\overline{\bm{p}}^{\text{DA}},\forall \hat{y}_{d,\tau} \in R_o^{\text{DA}}
\end{equation}

Likewise, the linear function between the optimal RT objective and $y_{d,\tau}-\hat{y}_{d,\tau} \in R_n$ can be obtained as follows,  
\begin{equation}\label{dualobjfunc}
   \gamma^{*n}_{d,\tau}(y_{d,\tau}-\hat{y}_{d,\tau})-\overline{\bm{\mu}}_{d,\tau}^{*n,\top}\overline{\bm{p}}^{-}-\overline{\bm{\eta}}_{d,\tau}^{*n,\top}\overline{\bm{p}}^{+},\forall y_{d,\tau}-\hat{y}_{d,\tau} \in R_n^{\text{RT}}
\end{equation}

With \eqref{DAdualobjfunc},\eqref{dualobjfunc}, the linear function between the forecast $\hat{y}_{d,\tau}$ and the overall operation cost defined in the partitions $R_o^{\text{DA}},R_n^{\text{RT}}$ can be obtained,
\begin{subequations}\label{overallfunc}
    \begin{align}
        &\eqref{DAdualobjfunc}+\eqref{dualobjfunc}\label{overallfunca}\\
        &\hat{y}_{d,\tau} \in R_o^{\text{DA}},y_{d,\tau}-\hat{y}_{d,\tau} \in R_n^{\text{RT}}
    \end{align}
\end{subequations}

Gathering the function defined in each DA partition $R_o^{\text{DA}}$ and each RT partition $R_n^{\text{RT}}$, the function between the forecast $\hat{y}_{d,\tau}$ and the overall operation cost is a piecewise linear function with $|\overline{\bm{p}}^{\text{DA}}| \times (|\overline{\bm{p}}^{+}|+|\overline{\bm{p}}^{-}|)$ segments. We use the symbol $\{\{\eqref{overallfunca}\}_{n=1}^{|\overline{\bm{p}}^{+}|+|\overline{\bm{p}}^{-}|}\}_{o=1}^{|\overline{\bm{p}}^{\text{DA}}|}$ to denote this function. We have the following proposition for the function $\{\{\eqref{overallfunca}\}_{n=1}^{|\overline{\bm{p}}^{+}|+|\overline{\bm{p}}^{-}|}\}_{o=1}^{|\overline{\bm{p}}^{\text{DA}}|}$
\begin{prop}
    The piecewise linear function between the forecast and the overall operation cost, i.e., $\{\{\eqref{overallfunca}\}_{n=1}^{|\overline{\bm{p}}^{+}|+|\overline{\bm{p}}^{-}|}\}_{o=1}^{|\overline{\bm{p}}^{\text{DA}}|}$, is a convex function.
\end{prop}
\begin{proof}
    With Theorem \ref{theom1}, $\{\eqref{DAdualobjfunc}\}_{o=1}^{|\overline{\bm{p}}^{\text{DA}}|}$ and $\{\eqref{dualobjfunc}\}_{n=1}^{|\overline{\bm{p}}^{+}|+|\overline{\bm{p}}^{-}|}$ are convex functions. $\{\{\eqref{overallfunca}\}_{n=1}^{|\overline{\bm{p}}^{+}|+|\overline{\bm{p}}^{-}|}\}_{o=1}^{|\overline{\bm{p}}^{\text{DA}}|}$ is the summation of two convex functions $\{\eqref{DAdualobjfunc}\}_{o=1}^{|\overline{\bm{p}}^{\text{DA}}|}$ and $\{\eqref{dualobjfunc}\}_{n=1}^{|\overline{\bm{p}}^{+}|+|\overline{\bm{p}}^{-}|}$. Therefore, it is convex.
\end{proof}

The input to the function $\{\{\eqref{overallfunca}\}_{n=1}^{|\overline{\bm{p}}^{+}|+|\overline{\bm{p}}^{-}|}\}_{o=1}^{|\overline{\bm{p}}^{\text{DA}}|}$ are the forecast $\hat{y}_{d,\tau}$ and the forecast deviation $y_{d,\tau}-\hat{y}_{d,\tau}$, while its output is the overall operation cost as stated in \eqref{overallcostb}. The function is illustrated in Fig. \ref{costillu}. In the next section, we will leverage the convexity of this function to perform the forecast model parameter estimation.

\begin{figure}[t]
  \centering
  % Requires \usepackage{graphicx}
  \includegraphics[scale=0.35]{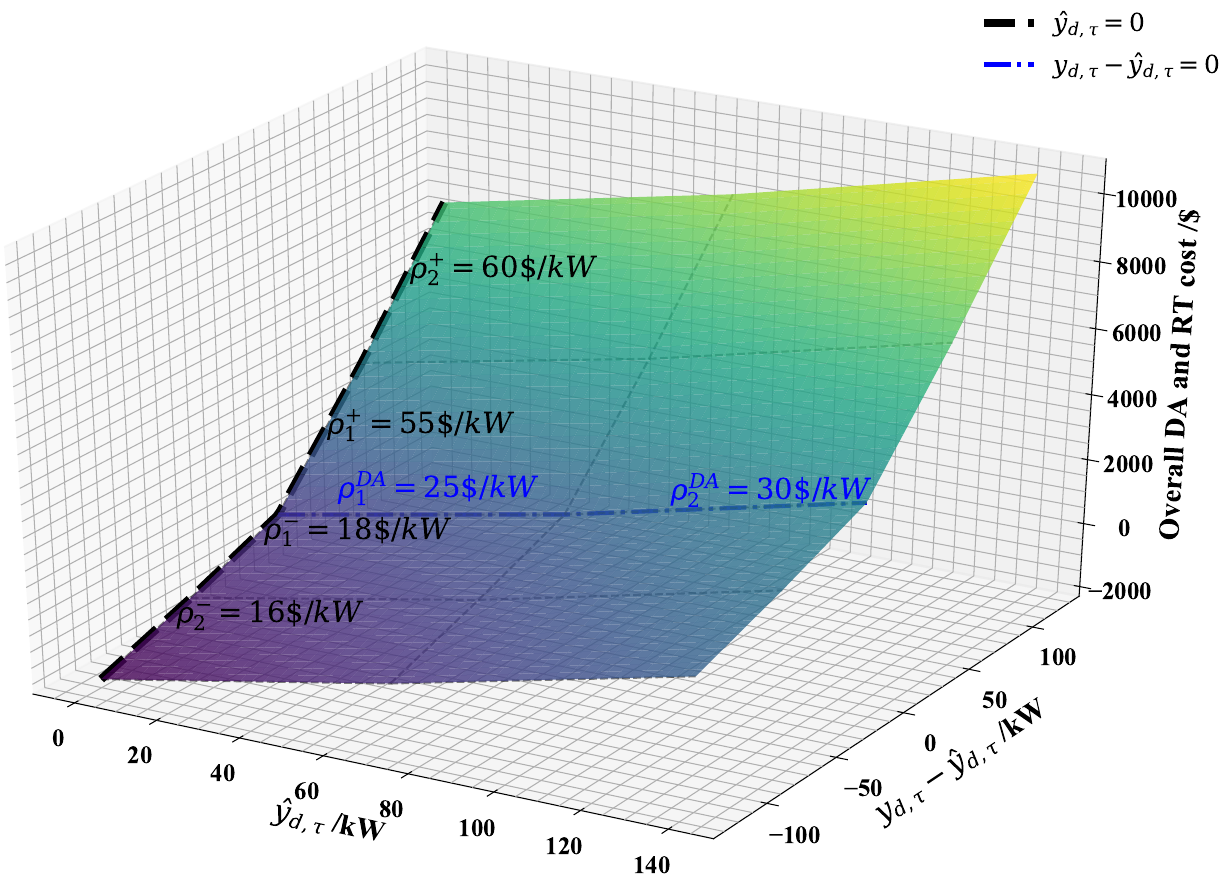}
  \caption{\scriptsize Illustration of the function between the overall operation cost and the forecast. In the DA operation, there are two generators with marginal costs $\rho_1^{\text{DA}}=25 \$/kW,\rho_2^{\text{DA}}=30 \$/kW$. In the RT operation, flexible resources 1 and 2 addressing energy deficit with marginal cost $\rho_1^{+}=55 \$/kW,\rho_2^{+}=60 \$/kW$. Flexible resources 3 and 4 address energy surplus with marginal utility $\rho_1^{-}=18 \$/kW,\rho_2^{-}=16 \$/kW$. The function has 8 segments.
  % The blue solid line represents the positive forecast deviation $y_{d.\tau}-\hat{y}_{d,\tau} > 0$. Since the flexible resource 1 is the cheapest one, it is dispatched to full capacity. The flexible resource 2 is partially dispatched until the deviation is settled. 
 }
\label{costillu}
\end{figure}

\section{Solution Strategy}

With the function $\{\{\eqref{overallfunca}\}_{n=1}^{|\overline{\bm{p}}^{+}|+|\overline{\bm{p}}^{-}|}\}_{o=1}^{|\overline{\bm{p}}^{\text{DA}}|}$, we can remove the lower level \eqref{Dbilevel_e}, and equivalently transform \eqref{Dbilevel} as,
\begin{subequations}\label{Sbilevel}
\begin{alignat}{2} 
&\underset{{\Theta},\alpha}{\mathop{\min}}&& \  \alpha+\nonumber\\
&&& \frac{1}{D\cdot T\cdot(1-\beta)}\sum_{d=1}^{D}\sum_{\tau=1}^{T}[\{\{\eqref{overallfunca}\}_{n=1}^{|\overline{\bm{p}}^{+}|+|\overline{\bm{p}}^{-}|}\}_{o=1}^{|\overline{\bm{p}}^{\text{DA}}|}-\alpha]^+\label{Sbilevel_a}\\
    & \text{s.t.} 
    && \eqref{bilevel_b}. \notag
\end{alignat}
\end{subequations}
Since $\{\{\eqref{overallfunca}\}_{n=1}^{|\overline{\bm{p}}^{+}|+|\overline{\bm{p}}^{-}|}\}_{o=1}^{|\overline{\bm{p}}^{\text{DA}}|}$ is a convex function, the objective \eqref{Sbilevel_a} is convex; see \cite{rockafellar2000optimization}. Here, we use the linear regression model in \eqref{bilevel_b} \cite{morales2023prescribing,dvorkin2024regression}. Although the model is linear in the context $\bm{s}_{d,\tau}$, the nonlinearity can be captured by the transformed context $\bm{s}_{d,\tau}$, which can be nonlinear kernel functions outputs  or neural networks hidden layers outputs. With the linear forecast model, the constraints \eqref{bilevel_b} are convex. The parameter estimation problem in \eqref{Sbilevel} is a convex program.

For each times-slot $\tau$, we can represent the piecewise linear and convex function $\{\{\eqref{overallfunca}\}_{n=1}^{|\overline{\bm{p}}^{+}|+|\overline{\bm{p}}^{-}|}\}_{o=1}^{|\overline{\bm{p}}^{\text{DA}}|}$ as the maximum of affine functions (where the number of affine functions is $|\overline{\bm{p}}^{\text{DA}}| \times (|\overline{\bm{p}}^{+}|+|\overline{\bm{p}}^{-}|)$); see \cite{boyd2004convex},
\begin{equation}\label{aux}
    \nu_{d,\tau} \geq \eqref{overallfunca}, \forall o=1,...,|\overline{\bm{p}}^{\text{DA}}|, \forall n=1,...,|\overline{\bm{p}}^{+}|+|\overline{\bm{p}}^{-}|
\end{equation}
where $\nu_{d,\tau}$ is the introduced auxiliary variable. The objective \eqref{Sbilevel_a} therefore becomes,
\begin{equation}\label{objt}
    \alpha+\frac{1}{D\cdot T\cdot(1-\beta)}\sum_{d=1}^{D}\sum_{\tau=1}^{T}[\nu_{d,\tau}-\alpha]^+
\end{equation}

Next, we introduce another auxiliary variable $\chi_{d,\tau} \geq 0$ to cope with $[\nu_{d,\tau}-\alpha]^+$, which takes the maximum between $\nu_{d,\tau}-\alpha$ and 0. The objective \eqref{objt} is further transformed as,
\begin{subequations}\label{objta}
\begin{align}
    &\alpha+\frac{1}{D\cdot T\cdot(1-\beta)}\sum_{d=1}^{D}\sum_{\tau=1}^{T}\chi_{d,\tau}\\
    & \text{s.t.}\ \chi_{d,\tau} \geq \nu_{d,\tau}-\alpha\label{objtab}\\
    & \ \chi_{d,\tau} \geq 0\label{objtac}
\end{align}
\end{subequations}
With \eqref{objta}, the parameter estimation \eqref{Sbilevel} can be equivalently rewritten as,
\begin{subequations}\label{Fbilevel}
\begin{alignat}{2} 
&\underset{{\Theta},\alpha}{\mathop{\min}}&&\ \alpha+\frac{1}{D\cdot T\cdot(1-\beta)}\sum_{d=1}^{D}\sum_{\tau=1}^{T}\chi_{d,\tau}\\
    & \text{s.t.} 
    && \eqref{bilevel_b},\eqref{objtab},\eqref{objtac},\eqref{aux}
\end{alignat}
\end{subequations}
which is a linear program, ready to be solved by off-the-shelf solvers. The global optimal solutions can therefore be ensured.

\section{Case Study}

We consider a VPP with 2 generators, whose energy dispatch needs to be settled in the DA operation, and 2 flexible resources for addressing the forecast deviation in the RT operation. Their technical data is provided in \cite{zhang2024repository}. The yearly net demand is used, which is the demand minus the wind power. To issue the forecast at time-slot $\tau$ on day $d$, the context consists of the net demand record at time-slot $\tau$ on the days $d-1,d-2,d-3$ and the record at time-slot $\tau+1$ on the day $d-2$, along with the numerical weather prediction (the estimated wind speed and direction at 10m and 100m altitude). The data can be found in \cite{zhang2024repository}.

Here, we use three benchmarks for comparison. The first one is trained to minimize mean squared error, which provides the prediction for the expected net demand. The second one is the value-oriented forecasting approach to minimize the expected operation cost, which is designed for a risk-neural operator. Lastly, the stochastic program to minimize $\text{CVaR}_{\beta}$, with 200 wind power scenarios obtained by k-nearest-
neighbors, is considered. The three benchmarks are denoted as \textbf{Qua-E}, \textbf{Val-N}, and \textbf{Sto-OPT}, respectively. The Qua-E, Val-N, and the proposed approach use the same linear forecast model. Sto-OPT serves as an ideal benchmark and is expected to yield the best results. The programs are solved by Gurobi.

We consider three evaluation metrics on the test set, i.e., Root Mean Squared Error (RMSE), average overall operation cost, and the average high operation cost above the $\beta$-quantile on the test set. Let $D_{\text{test}}$ be the number of days on the test set. The average overall operation cost is,
\begin{equation}\label{aop}
    \frac{1}{D_{\text{test}} \cdot T}\sum_{d=1}^{D_{\text{test}}}\sum_{\tau=1}^T\eqref{overallcosta}
\end{equation}

Let $\{\eqref{overallcosta}|\eqref{overallcosta} > \beta\text{-quantile}\}$ be the set of the overall costs above the $\beta$-quantile, whose cardinality is denoted as $M_{abov}$. The average high operation cost above the $\beta$-quantile is,
\begin{equation}\label{abo}
    \frac{1}{M_{abov}}\sum \{\eqref{overallcosta}|\eqref{overallcosta} > \beta\text{-quantile}\}
\end{equation}

\subsection{Operational Advantage}

We investigate the performance of our approach and the three benchmarks. The RMSE, average operation cost \eqref{aop}, and average high operation cost above the $\beta$-quantile \eqref{abo}, along with the computation time on the test/training sets are shown in Table \ref{tab1}. The $\beta$ is set as 0.5. The results demonstrate that the proposed approach achieves the lowest cost measured by \eqref{abo} (which represents the average of high operation costs) among Qua-E and Val-N, and is very close to the ideal benchmark Sto-OPT. Compared to Val-N, the cost reduction measured by \eqref{abo} reaches 1.9\%. Additionally, the proposed approach has less computation time on the test set, compared to Sto-OPT, which demonstrates the computational efficiency. Also, the training time only takes 8 \unit{s}, since the linear program \eqref{Fbilevel} is efficient to solve. The expected operation cost \eqref{aop} resulting from the proposed approach is in the middle of Val-N and Qua-E. This is reasonable since Qua-E completely ignores the decision value during training, while Val-N uses the expected operation cost \eqref{aop} as the training objective. The metric RMSE indicates that the proposed approach achieves the worst accuracy. This can be further demonstrated by the forecast profiles in Fig. \ref{profile}. The proposed approach tends to forecast more net demand, to avoid less favorable case of energy deficit, when the marginal cost of flexible resources addressing energy deficit in RT is larger than the marginal utilities of flexible resources addressing energy surplus.

\begin{figure}[ht]
  \centering
  % Requires \usepackage{graphicx}
  \includegraphics[scale=0.5]{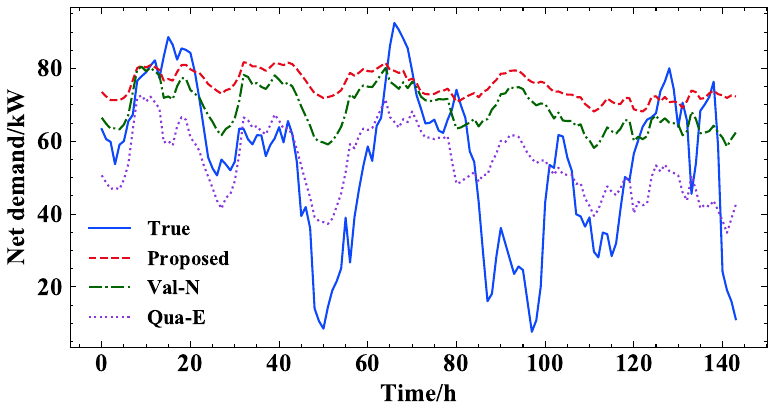}\\
  \caption{The 6-day net demand forecast profiles, when the marginal cost of flexible resources addressing energy deficit in RT is larger than the marginal utilities of flexible resources addressing energy surplus. $\beta$ is set as 0.5.}\label{profile}
\end{figure}

\begin{table}
% \small
\scriptsize
% \footnotesize
\caption{ RMSE, average operation cost \eqref{aop}, average high operation cost  \eqref{abo} and the test/training time of the proposed approach and benchmarks.}\label{tab1} 

\begin{center}
\begin{tabular}{ c  c  c  c c }
\hline\hline
         &\makecell[c]{Proposed} & \makecell[c]{Qua-E} & \makecell[c]{Val-N} &  \makecell[c]{Sto-OPT}\\
\hline
    RMSE (\unit{kW}) & 23 & 19 & 17  & -\\
    Average cost \eqref{aop} (\unit{\$}) & 1762 & 1880 & 1752 & 1754\\
    Average high cost \eqref{abo} (\unit{\$}) & 2111 & 2430 & 2152 & 2102\\
    \makecell[c]{Test time} (\unit{s}) & 2  & 2  & 2  & 60\\
    \makecell[c]{Training time} (\unit{s}) & 8  & 4  & 8  & -\\
\hline\hline
\end{tabular}
\end{center}
\vspace{-2em}
\end{table}

Furthermore, we compare the average high operation cost \eqref{abo} under different risk-averse levels, which is controlled by $\beta$. When $\beta$ is larger, the operator is more averse to the high costs. The proposed approach achieves the least average high cost among the benchmarks when $\beta$ varies. The cost reduction is more obvious when $\beta$ takes a larger value.

\begin{table}
\scriptsize
\caption{ Average high operation cost  \eqref{abo} of the proposed approach and benchmarks. The cost reduction of the proposed approach compared to Val-N is listed in the last column.}\label{tab1}
\begin{center}
\begin{tabular}{ c  c  c  c c}
\hline\hline
         &\makecell[c]{Proposed} & \makecell[c]{Qua-E} & \makecell[c]{Val-N} & Reduction\\
\hline
    $\beta$=0.3 & 1984 \$ & 2215 \$ & 1997 \$  & 0.7\%\\
    $\beta$=0.5 & 2111 \$ & 2430 \$ & 2152 \$ & 1.9\%\\
    $\beta$=0.7 & 2248 \$ & 2678 \$ & 2362 \$  & 4.8\%\\
\hline\hline
\end{tabular}
\end{center}
\vspace{-2em}
\end{table}

\section{Conclusions}
This paper explains how virtual power plants can mitigate the risk of high costs via risk-aware value-oriented net demand forecasting. We propose a risk-aware parameter estimation approach for training a forecast model. The objective is to minimize the $\text{CVar}_{\beta}$ of the overall operation cost, which is a measure of the average high operation cost. Leveraging the structure of the operation problem, we show the function between the overall operation cost and the forecast is a convex function. The parameter estimation can be transformed into a convex program when the forecast model is linear. Compared to the value-oriented forecasts from the perspective of a risk-neutral operator, our approach reduces the average high operation costs by 0.7\%-4.8\%. Future works can be developed by incorporating nonlinear and nonconvex forecast models.

\vspace{-0.5em}

\bibliographystyle{IEEEtran}
\bibliography{IEEEabrv,mylib}
% \begin{thebibliography}{00}
% \bibitem{b1} G. Eason, B. Noble, and I. N. Sneddon, ``On certain integrals of Lipschitz-Hankel type involving products of Bessel functions,'' Phil. Trans. Roy. Soc. London, vol. A247, pp. 529--551, April 1955.
% \bibitem{b2} J. Clerk Maxwell, A Treatise on Electricity and Magnetism, 3rd ed., vol. 2. Oxford: Clarendon, 1892, pp.68--73.
% \bibitem{b3} I. S. Jacobs and C. P. Bean, ``Fine particles, thin films and exchange anisotropy,'' in Magnetism, vol. III, G. T. Rado and H. Suhl, Eds. New York: Academic, 1963, pp. 271--350.
% \bibitem{b4} K. Elissa, ``Title of paper if known,'' unpublished.
% \bibitem{b5} R. Nicole, ``Title of paper with only first word capitalized,'' J. Name Stand. Abbrev., in press.
% \bibitem{b6} Y. Yorozu, M. Hirano, K. Oka, and Y. Tagawa, ``Electron spectroscopy studies on magneto-optical media and plastic substrate interface,'' IEEE Transl. J. Magn. Japan, vol. 2, pp. 740--741, August 1987 [Digests 9th Annual Conf. Magnetics Japan, p. 301, 1982].
% \bibitem{b7} M. Young, The Technical Writer's Handbook. Mill Valley, CA: University Science, 1989.
% \end{thebibliography}

\end{document}